\documentclass[submission,copyright,creativecommons]{eptcs}
\providecommand{\event}{FMAS 2024}\usepackage[T1]{fontenc}
\usepackage{listings}
\usepackage{multicol}
\usepackage{amsmath}
\usepackage{subcaption}
\usepackage{listings}
\usepackage{wrapfig}
\usepackage{multirow}
\usepackage{booktabs}

\usepackage{amsthm}
\theoremstyle{definition}
\newtheorem{example}{Example}[section]
\newtheorem{theorem}{Theorem}[section]
\newtheorem{lemma}{Lemma}[section]

\usepackage{graphicx} 
\usepackage{url}
\usepackage{booktabs}
\usepackage{amssymb}
\usepackage{xspace}
\usepackage{listings}
\usepackage{relsize}

\usepackage{hyperref}
\hypersetup{%
  colorlinks = true,
  linkcolor  = blue!75!black,
  citecolor = blue!75!black
}
\usepackage[capitalise,nameinlink]{cleveref}

\newcommand{\wrap}[1]{\begin{tabular}{@{}c@{}}#1\end{tabular}}

\usepackage[normalem]{ulem}
\usepackage[textwidth=24.0mm,textsize=scriptsize]{todonotes}
\newcommand{\bgtext}[1]{%
  \bgroup\markoverwith {\textcolor{#1}{\rule[-0.5ex]{2pt}{11pt}}}\ULon}

\lstdefinestyle{scala}{
    keywordstyle=\ttfamily\color{blue},
    basicstyle=\scriptsize\ttfamily,
    showspaces=false, 
    identifierstyle=\ttfamily\color{black}\bfseries, 
    commentstyle=\color{gray},
    stringstyle=\ttfamily,
    morekeywords={if,then,else,for,while,do}
}
\lstdefinestyle{comments}{
    language=scala,
    basicstyle=\scriptsize\ttfamily,
    keywordstyle=\ttfamily\color{black},
    identifierstyle=\ttfamily\color{black}\bfseries, 
    commentstyle=\color{black},
    stringstyle=\ttfamily,
    breakatwhitespace=true,
    breaklines=true,
}
\lstdefinestyle{lince}{
    keywordstyle=\ttfamily\color{blue}\bfseries,
    basicstyle=\scriptsize\ttfamily,
    showspaces=false, 
    identifierstyle=\ttfamily\color{black}\bfseries,
    morecomment=[l]{//},
    commentstyle=\color{gray},
    rulecolor=\color{black!40},         
    xleftmargin=1.5mm,
    xrightmargin=1.5mm,
    backgroundcolor=\color{black!5},
    numbersep=5pt, 
    numberstyle=\tiny\color{gray},   
    stepnumber=1,  
    captionpos=b, 
    belowcaptionskip=5mm,
    frame=single,                    
    morekeywords={if,if,then,else,for,while,do,repeat,until}
    stringstyle=\ttfamily\color{green!50!black},
    showspaces=false,                
    showstringspaces=false,
    showtabs=false,
    breakatwhitespace=true,
    breaklines=true,
    literate=*
             {==}{{\,==\,}}{3}
             {'=}{{'{\color{blue}\,\!=}}}{2}
             {:=}{{{\color{blue}:=}}}{2}
             {||}{{{\color{red!70!black}$\rule{1mm}{0mm}\|\rule{1mm}{0mm}$}}}{1}
}
\lstdefinestyle{error}{
    basicstyle=\scriptsize\sf\color{red!65!black},
    backgroundcolor=\color{red!10},
    xleftmargin=4mm,
    xrightmargin=4mm,
    framexleftmargin=4mm,
    framexrightmargin=4mm,    
    breakatwhitespace=true,
    breaklines=true,
}

\newcommand{\code}[1]{{\relsize{-1}\ttfamily #1}}

\renewcommand{\>}{\hspace*{5pt}}

\newcommand{\blue}[1]{\textcolor{blue}{#1}}
\newcommand{\prog}[1]{\ensuremath{\tt #1}\xspace}
\newcommand{\until}{\textcolor{blue}{\prog{for}}\xspace}

\newcommand{\progife}[3]{{ \blue{\prog if} \> #1 \> \blue{\prog then} \> 
{\prog #2} \> \blue{\prog else} \> {\prog #3}}}
\newcommand{\progwhile}[2]{{\blue{\prog while} \>  #1 \> \blue{\prog do} \> \{ \> {#2}  \> \}}}
\newcommand{\scomp}{\, \blue{;} \,}

\newcommand{\Reals}{\mathbb{R}}

\newcommand{\Rz}{\Reals_{\geq 0}}

\newcommand{\sem}[1]{\left \llbracket #1 \right \rrbracket}
\newcommand{\lrule}[3]{\textbf{#1}\quad\frac{#2}{#3}}
\newcommand{\bsto}{~\Downarrow~}
\newcommand{\ssto}[1][]{~\to^{#1}~}
\newcommand{\stp}{\mathit{stop}}
\newcommand{\skp}{\mathit{skip}}
\newcommand{\err}{\mathit{err}}
\providecommand{\comma}{,\operatorname{}\linebreak[1]}		\newcommand{\sep}{\kern1pt\comma\kern-1pt}
\usepackage{stmaryrd}
\usepackage{mathtools}
\usepackage{xfrac}
\newtheorem{proposition}{Proposition}
\newtheorem{corollary}{Corollary}

\newcommand{\titlerunning}{Formal Simulation and Visualisation of Hybrid Programs}
\newcommand{\authorrunning}{R. Correia, P. Mendes, R. Neves, J. Proença}

\hypersetup{
  bookmarksnumbered,
  pdftitle    = {\titlerunning},
  pdfauthor   = {\authorrunning},
  pdfsubject  = {Lince},}
\title{Formal Simulation and Visualisation of Hybrid Programs
\\ {\normalsize An Extension of a Proof-of-Concept Tool } } 
\author{
Pedro Mendes
\institute{University of Minho, Portugal}
\email{pg50685@alunos.uminho.pt}
\and
Ricardo Correia
\institute{University of Minho, Portugal}
\email{pg47607@alunos.uminho.pt}
\and
Renato Neves
\institute{INESC-TEC \& University of Minho, Portugal}
\email{nevrenato@di.uminho.pt}
\and
José Proença
\institute{CISTER, Faculty of Sciences of the University of Porto, Portugal}
\email{jose.proenca@fc.up.pt}
}

\begin{document}
\maketitle\begin{abstract}

The design and analysis of systems that combine computational
behaviour with physical processes' continuous dynamics -- such as movement,
velocity, and voltage -- is a famous, challenging task.  Several theoretical
results from programming theory emerged in the last decades to tackle the
issue; some of which are the basis of a \emph{proof-of-concept} tool, called
Lince, that aids in the analysis of such systems, by presenting simulations of
their respective behaviours.

However being a proof-of-concept, the tool is quite limited with respect to  usability,
and when attempting to apply it to a set of common, concrete problems,
involving autonomous driving and others, it either simply cannot simulate them
or fails to provide a satisfactory user-experience.

The current work complements the aforementioned theoretical approaches with a
more practical perspective, by improving Lince along several dimensions: to name
a few, richer syntactic constructs, more operations, more informative plotting
systems and errors messages, and a better performance overall.
We illustrate our improvements via a variety of examples that involve both
autonomous driving and electrical systems.

\end{abstract}

\section{Introduction}
\label{Introduction}
\noindent \textbf{Motivation and context.} This paper concerns the design and
analysis of hybrid systems (\emph{i.e.} those that combine discrete with continuous
behaviour) from a programming-oriented perspective. Such a view emerged recently
in a series of
works~\cite{Platzer10,neves18,goncharov2020implementing,kamburjan22}, and
revolves around the idea of importing principles and techniques from
programming theory to better handle the behaviour of hybrid
systems. In this context programs combine standard program constructs, such as
conditionals and while-loops, with certain kinds of differential statement
meant to express the dynamics of physical processes, such as time, energy, and
motion. Consider the following example of such a program:
\begin{align}
\label{eq:mot}
& p' = v, v' = 2 \> {\prog{\blue{for}}} \> 1 \> \scomp
\> p' = v, v' = -2 \> {\prog{\blue{for}}} \> 1  
\end{align}

\begin{wrapfigure}[7]{r}{0.4\textwidth} 
~\\[-4mm]
\hphantom{1cm} \includegraphics[width=0.35\textwidth]{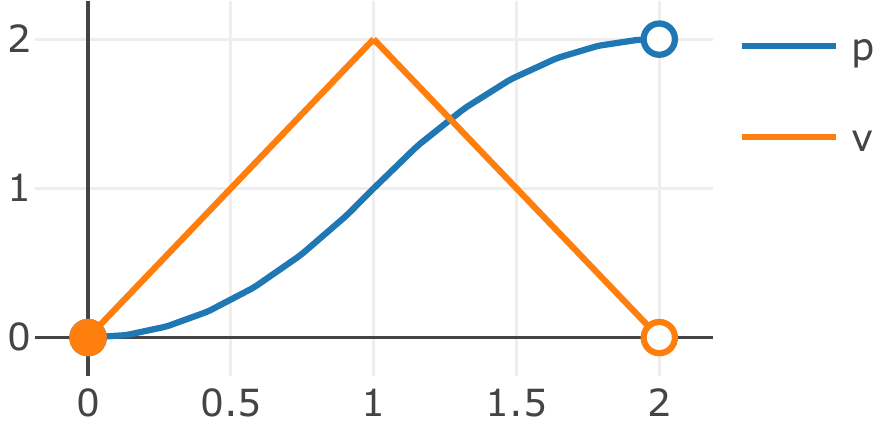}
\caption{Simulation of \eqref{eq:mot}.}
\label{fig:mot}
\end{wrapfigure}
In a nutshell, it is a sequential composition $(\scomp)$ of two 
programs where each expresses how the position $(p)$ and velocity $(v)$ of a
vehicle evolve over time. The program on the left ($p' = v, v' = 2 \>
{\prog{\blue{for}}} \> 1$) is a differential statement that reads {``the
vehicle accelerates at the rate of $2\,\sfrac{m}{s^2}$ for 1 second''}.
The other program corresponds to a deceleration. Both position and
velocity over time  are presented in \cref{fig:mot}, where we see that the
vehicle travels 2 meters and then stops.

Actually there has been a rapid proliferation of such systems, not only in the
domain of autonomous driving but also in the medical industry and industrial
infrastructures, among others~\cite{Platzer10,gunes14,lee16,neves18}. This
spurred extensive research on languages, semantics, and tools for their design
and analysis. An example is our
work~\cite{goncharov2019,goncharov2020implementing} on the semantics of hybrid
programs -- \emph{i.e.} those that combine program constructs with differential
statements, such as in~\eqref{eq:mot} -- from which arises a mathematical basis
for reasoning about their behaviour, both operationally and denotationally. A
proof-of-concept tool, called Lince, also emerged from this: its engine is a
previously developed operational semantics~\cite{goncharov2020implementing}
that yields trajectories of hybrid programs, just as we saw in \cref{fig:mot}.
However because our focus was rather theoretical, the tool was not developed
with usability in mind, and thus lacks basic features for tackling a broad
range of important scenarios. Let us illustrate this problem with a very simple
example.

\smallskip
\noindent
\textbf{Problem scenario.} Suppose that we wish to move a stationary object a
distance of $dist$ meters -- a basic task in autonomous driving. For simplicity
assume that we have access only to the acceleration rates $a\,\sfrac{m}{s^2}$ and $-a\,\sfrac{m}{s^2}$, where $a > 0$. Our mission can be accomplished by taking the following
variation of~\cref{eq:mot}, 
\begin{align}
& p' = v, v' = a \> {\prog{\blue{for}}} \> t \> \scomp
\> p' = v, v' = -a \> {\prog{\blue{for}}} \> t  
\end{align}
for a suitable duration~$t$. Then in order to calculate $t$ (\emph{i.e.} the
prescribed duration of each differential statement) we simply observe that,
\[
        dist = \int^t_0 v_a(x) \, d x \> + \> \int^t_0 v_{-a}(x) \, dx 
\]
where $v_a(x) = a \cdot x$ and $v_{-a}(x) = v_a(t) + -a \cdot x$ are the
velocity functions with respect to the time intervals $[0,t]$ and $[t,2\cdot
t]$ associated with the program's execution. We now observe, by recalling
\cref{fig:mot}, that the value $dist$ corresponds to the area of a triangle
with basis $2 \cdot t$ and height $v_a(t)$. This geometric shape yields the
equations,
\[
        \begin{cases}
                dist & = \sfrac{1}{2} \cdot (2 \cdot t) \cdot v_a(t) 
                \hspace{0.2cm} \text{\textbf{(area)}} \\
                v_a(t) & = a \cdot t \hspace{0.2cm} \text{\textbf{(height)}} 
        \end{cases}
        \hspace{0.5cm} \Longrightarrow \hspace{0.5cm}
        t = \sqrt{ \frac{dist}{a} }
\]
Finally observe that if $dist = 3$ and $a = 1$ then $t= \sqrt{3}$.
Unfortunately the previous version of Lince does not support square root
operations which renders our mission impossible to accomplish.

\smallskip
\noindent
\textbf{Contributions and outline.}
As already alluded to, this paper complements our previous theoretical work on
the semantics of hybrid
programming~\cite{goncharov2019,goncharov2020implementing}. Specifically it
improves our proof-of-concept tool Lince so that it can handle a myriad of
important scenarios, whilst maintaining both its simplicity and theoretical
underpinnings.  The improvements were made along different dimensions, and we
highlight the most relevant ones next\footnote{The improved version can be
        checked online at \url{http://arcatools.org/lince}.}.

\smallskip
\noindent
\emph{Extension of basic operations.} As illustrated before, the previous version of
Lince lacked essential arithmetic operations for handling most basic tasks.
Thus as the first main contribution we added standard arithmetic operations,
including divisions, trigonometric functions, and square root extractions.
Notably the fact that many of these operations are partial required us to
extend the operational semantics developed in~\cite{goncharov2020implementing}
(the main engine of Lince) with the possibility of failure.  The extended
semantics is detailed in \cref{sec:modelling} and it is of course the basis of
the new engine behind improved Lince. 

\smallskip
\noindent
\emph{Extension of numerical methods.} Again because our focus in previous work was
rather theoretical the previous version of Lince was unable to simulate
standard scenarios in hybrid programming. A main reason for this was our method
of obtaining solutions of systems of ordinary differential equations (ODEs),
which although \emph{exact} lacked in scalability. Precisely for this reason we
now integrate a complementary, numerical solver in Lince with the obvious
compromise that the solutions obtained for such systems are no longer exact.

The benefits of the extended language (and respective semantics), the numerical
solver, and a number of quality-of-life features, are summarised
in~\cref{sec:interpreter} and illustrated with a standard, running example
concerning the famous concept of harmonic oscillation.

\smallskip
\noindent
\emph{Extension of visualisation mechanisms.} Lince is constituted by two core
components: the simulator which, by recurring to the aforementioned operational
semantics, parses a received program and presents its output with respect to a
\emph{single} time instant. And the visualiser which presents (a sample of) the
trajectory over time respective to the program at hand, by querying the
operational semantics for a certain sequence of time instants.  After trying to
properly visualise the behaviour of several types of hybrid program with Lince
we identified two major limitations with respect to this architecture. First many
real-world problems involve multiple spatial dimensions and thus the described
view of trajectories over time is often not the best representation of 
a hybrid program's behaviour. Second the user is often interested in
observing the overall program behaviour for varying initial conditions,
concerning for example position and velocity. We therefore present in
\cref{sec:visualiser} an improved visualiser for Lince that precisely addresses
these two issues. We illustrate it via another classical scenario in autonomous
driving, \emph{viz.} manoeuvring around an obstacle.

In \cref{sec:pursuit} we illustrate that, whilst keeping its simplicity, Lince
can now handle complex central problems in the theory of hybrid systems; we
focus specifically on the task of one player pursuing another, \emph{e.g.} a
vehicle, a drone, or simply a projectile. Such pursuit games were discussed for
example in~\cite{manna93,anderson93,chaochen93,krilavicius05}, from an
(hybrid-)automata, state-chart, and duration calculus perspective.  Here we
present a programming-oriented approach. Finally in \cref{sec:conc} we discuss
future work and conclude.

\smallskip
\noindent
\textbf{Related work.}
Several tools for the design and analysis of hybrid systems were already
proposed, \emph{e.g.} in the areas of deductive verification~\cite{Platzer10},
model checking~\cite{ballarini11,frehse11,bresolin15},
simulation~\cite{klee2007,fritzson2014,kamburjan22,goncharov2020implementing},
and program semantics~\cite{Platzer10, kamburjan22, goncharov2020implementing}.
But only a few are committed to a programming-oriented approach, rooted on
formal semantics, and with effective simulation capabilities. The only ones we
are aware of are \cite{kamburjan22} and our own tool
Lince~\cite{goncharov2020implementing}.  Interestingly both cases adopt
complementary approaches as well: the former harbours a very powerful
concurrent language, particularly well-suited for large-scale distributed
systems. The latter, harbouring a sequential while-language, aims at being
minimalistic  whilst still capturing a broad range of interesting problems on
which to study different aspects of (pure) hybrid computation at a suitable
abstraction level.

Aside from the obvious pedagogical benefit, our minimalistic approach also
allows to capitalise on different programming theories more easily.  For
example already in~\cite{goncharov2020implementing}  we connected our tool to a
compositional, denotational semantics -- particularly well-suited to study
hybrid program equivalence and combinations with other paradigms. An analogous
concurrent semantics for~\cite{kamburjan22} would be notoriously more difficult
to achieve (\emph{cf.}~\cite{reynolds98,winskel93}). Similarly our language is
amenable to algebraic reasoning in the style of (weak) Kleene
algebras~\cite{kozen97,hofner_phdthesis} whilst the connection between the
latter and concurrent object-oriented programming (as adopted
in~\cite{kamburjan22}) is less clear. 

\section{Lince's Foundations Extended with the Possibility of Failure}
\label{sec:modelling}

We now extend part of Lince's foundations with the possibility of failure.
Specifically we present an extension of the language
in~\cite{goncharov2020implementing} with partial operations, such as division
and square root extraction, and introduce a corresponding operational
semantics. As explained in the introduction, such is necessary for extending
Lince to `real-world problems' whilst preserving its merit of having a firm,
mathematical basis. 

\smallskip
\noindent
\textbf{Language.}
First we postulate a finite set $X = \{ x_1, \dots, x_n \} $ of variables and a
stock of partial functions $f : \Reals^n \xrightharpoonup{\hspace{0.1cm}}
\Reals$ that contains the usual arithmetic operations. Then we define
expressions and boolean conditions via the following BNF grammars,
\begin{align*}
        e ::= x \mid f(e,\dots,e)
        \hspace{2cm}
        b ::= e \leq e \mid b \wedge b \mid
        b \vee b \mid \neg\, b \mid \prog{tt} \mid \prog{ff}
\end{align*}
We omit the explanation of these grammars as they are widely used (see
\emph{e.g.}~\cite{winskel93,reynolds98}).  Next, we qualify as `linear' those
expressions $e$ which aside from the use of variables involve only the
operations $+$ and $r \cdot (-)$ for some $r \in \Reals$.  For example the
expression $2 \cdot x$ is linear but the expression $x \cdot x$ is not.  The
concept of linearity is key in the grammar of hybrid programs which we present
next.

Programs are built according to the following BNF grammars,
\begin{align*}
        \prog{a} & ::= x'_1 = \ell_1, \dots, x'_n = \ell_n \> \until  \> e
        \mid x := e
        \\
        \prog{p} & ::= {\prog a} \mid {\prog p} \scomp {\prog p} \mid \progife{b}{p}{p}
        \mid \progwhile{b}{{\prog p}}
\end{align*}
where the terms $\ell_i$ ($1 \leq i \leq n$) are linear expressions. We qualify
as `atomic' those hybrid programs that are built according to the first
grammar. They can be either classical assignments or \emph{differential}
statements as described in the introduction. The linearity constraint is here
imposed merely to ensure that the latter kind of statement will always have
unique solutions, which renders our semantics more lightweight whilst still
being able to treat a broad range of problems (see more details
in~\cite{goncharov2020implementing}).

The language of hybrid programs $\prog{p}$ itself is simply the usual
while-language \cite{winskel93,reynolds98} extended with the aforementioned
differential statements.  It is easy to check that our grammar indeed extends
that in the previous version of Lince~\cite{goncharov2020implementing} where
\emph{all} expressions involved in the assignments and the durations of
differential statements had to be linear.  This has of course significant
implications in the operational semantics introduced
in~\cite{goncharov2020implementing}.

\smallskip
\noindent
\textbf{Operational semantics.}
We need a series of preliminaries. First for simplicity we abbreviate
differential statements $x'_1 = \ell_1, \dots, x'_n = \ell_n \> \until \> e$
simply to $\vec{x}' = \vec{\ell} ~ \until ~ e$, where $\vec{x}'$ and
$\vec{\ell}$ abbreviate the corresponding vectors of variables $x_1' \dots
x_n'$ and linear expressions $\ell_1 \dots \ell_n$. We call functions of the
type $\sigma : X \to \Reals$ \emph{environments}; they map variables to the
respective valuations.  We use the notation $\sigma[\vec x \mapsto \vec v]$ to
denote the environment that maps each $x_i$ in $\vec x$ to $v_i$ in $\vec{v}$
and the remaining variables as in $\sigma$. Finally we denote by
$\phi_\sigma^{\vec{x}' = \vec{\ell}}:  \Rz \to \Reals^n$ the (unique) solution
of a system of differential equations $\vec{x}' = \vec{\ell}$ with $\sigma$ as
the initial condition (recall our previous constraint about linearity which
ensures that such solutions indeed exist). When clear from context, we omit
both the superscript and subscript in $\phi_\sigma^{\vec{x}' = \vec{\ell}}$.
Next, for an expression $e$ the notation $\sem{e}(\sigma)$ denotes the standard
(partial) interpretation of expressions~\cite{winskel93,reynolds98} according
to $\sigma$, and analogously for $\sem{b}(\sigma)$ where $b$ is a boolean
expression. For example $\sem{x + 1}(\sigma) = \sigma(x) + 1$ and
$\sem{\sfrac{1}{x}}(\sigma)$ is undefined if $\sigma(x) = 0$. 

We now present an operational semantics for the language. Following traditions
in programming theory~\cite{olderog92, winskel93,reynolds98}, we present it
from two different, complementary perspectives, which gives a much more
complete understanding of the language's features. Specifically we present the
semantics in two different styles: one formalises the idea of a machine
``running'' a hybrid program and describes its \emph{step-by-step evolution}.
The other abstracts away from all \emph{intermediate steps} of this machine and
is therefore generally more suitable to reason about ``input-output
behaviours'' (although we do not explore such a feature here).  Whilst the
former style is the basis of Lince's new version, the latter style is
conceptually more intuitive and therefore we present it first. The current
section concludes with a proof that both semantics are in fact equivalent.  The
curious reader can consult for example~\cite{winskel93,reynolds98} for a
thorough account on the key differences between the small-step and big-step
styles in general program semantics.

Our `big-step' operational semantics is given by an `input-output' relation
$\Downarrow$ which relates programs ${\prog p}$, environments $\sigma$, and
time instants $t$ to outputs $v$. The expression ${\prog p} \sep \sigma \sep t
\bsto v$ can be read as ``at time instant $t$ the program {\prog p} starting from 
state $\sigma$ outputs $v$''. The relation $\Downarrow$ is built inductively
according to the rules in \cref{big_step}.  Specifically the first three rules
describe how differential statements are evaluated: first one computes the
duration $\sem{e}(\sigma)$ of the differential statement at hand and an error
is raised if $\sem{e}(\sigma)$ is undefined; otherwise the output $v$ is the
respective modified state (as dictated by the differential statement) paired
with one of the flags $\stp$ or $\skp$.  Intuitively the flag $\stp$ indicates
that we `reached' the time instant at which the program needs to be evaluated
and therefore the evaluation can stop moving forward in time, which fact is
reflected in rule $(\textbf{seq-stop})$. The flag $\skp$ is simply the negation
of $\stp$. The remaining rules follow analogous principles and therefore we
refrain from detailing them -- instead we will briefly show how the semantics
works via instructive, concrete examples.
\begin{figure}[h]
\begin{minipage}{1\textwidth}
\begin{flalign*} 
        \lrule{(diff-skip)}{
                \sem{e}(\sigma) = t
        }{
                \vec{x}' = \vec{\ell} \>
                \until \> e \sep\sigma \sep t 
                \bsto
                \skp\sep\sigma[\vec{x} \mapsto \phi(t)]
        }
\end{flalign*} 
\vspace{-4mm}
\begin{flalign*}
        \lrule{(diff-stop)}{
                \sem{e}(\sigma) > t
        }{
                \vec{x}' = \vec{\ell} \>
                \until \> e \sep\sigma\sep t
                \bsto
                \stp\sep\sigma[\vec{x} \mapsto \phi(t)]
        }
&&
\lrule{(diff-err)}{\sem{e}(\sigma) \text{ undefined }}{
  \vec{x}' = \vec{\ell} \>
  \until \> e \sep\sigma\sep t
    \bsto
    \err
}
\end{flalign*} 
\vspace{-4mm}
\begin{flalign*}
        \lrule{(asg-skip)}{
                \sem{e}(\sigma) \text{ defined }
        }{
                x := e\sep\sigma\sep 0
                \bsto
                \skp\sep\sigma[x \mapsto \sem{e}(\sigma)]
        }
&&
        \lrule{(asg-err)}{ 
                \sem{e}(\sigma) \text{ undefined }
        }{
                x := e\sep\sigma\sep t
                \bsto
                \err
        }
\end{flalign*} 
\vspace{-4mm}
\begin{flalign*}
        \lrule{(seq-skip)}{
                {\prog p}\sep\sigma\sep t \bsto \skp\sep\tau \qquad 
                {\prog q}\sep\tau \sep u \bsto v
        }{
                {\prog p}\scomp {\prog q}
                \sep\sigma\sep t + u \bsto v
        }
\end{flalign*} 
\vspace{-4mm}
\begin{flalign*}
        \lrule{(seq-stop)}{
                {\prog p}\sep\sigma\sep t \bsto \stp\sep\tau
        }{
                {\prog p}\scomp {\prog q}\sep\sigma\sep 
                t  \bsto \stp\sep\tau
        }
&&
        \lrule{(seq-err)}{{\prog p}\sep\sigma\sep t \bsto \err
        }{
                {\prog p}\scomp {\prog q}\sep\sigma
                \sep t  \bsto \err
        }
\end{flalign*} 
\vspace{-4mm}
\begin{flalign*}
        \lrule{(if-true)}{
                \sem{{b}}(\sigma)=\mathtt{tt}
                \qquad {\prog p} \sep\sigma\sep t\bsto v
        }{
                { \progife{b}{p}{q} } \sep\sigma\sep t \bsto v
        }
\end{flalign*} 
\vspace{-4mm}
\begin{flalign*}
        \lrule{(if-false)}{
                \sem{b}(\sigma)=\mathtt{ff}\qquad {\prog q} 
                \sep\sigma\sep t\bsto v
        }{
                \progife{b}{p}{q}\sep\sigma\sep t \bsto v
        }
&&
        \lrule{(if-err)}{
                \sem{b}(\sigma) \text{ undefined }
        }{
                \progife{b}{p}{q}\sep\sigma\sep t \bsto \err
        }
\end{flalign*} 
\vspace{-4mm}
\begin{flalign*}
        \lrule{(wh-true)}{
                \sem{b}(\sigma) = \mathtt{tt}\qquad {\prog p} 
                \scomp \progwhile{b}{p}\sep\sigma\sep t\bsto v 
        }{
                \progwhile{b}{{\prog p}}\sep\sigma\sep t 
                \bsto 
                v
        }
\end{flalign*}
\vspace{-4mm}  
\begin{flalign*}
        \lrule{(wh-false)}{
                \sem{b}(\sigma) = \mathtt{ff}
        }{
                \progwhile{b}{{\prog p}}\sep\sigma\sep 0
                \bsto 
                \skp\sep\sigma 
        }
&&
        \lrule{(wh-err)}{
                \sem{b}(\sigma) \text{ undefined }
        }{
                \progwhile{b}{{\prog p}}\sep\sigma\sep t
                \bsto 
                \err
        }
\end{flalign*}
\end{minipage}
\caption{Extension of the big-step operational semantics 
in~\cite{goncharov2020implementing} with the possibility of failure.}
\label{big_step}
\end{figure}

\begin{example}
        Let us consider the following very simple program, 
        \begin{align*}
        & x' = -1 \> {\prog{\blue{for}}} \> 1 \> \scomp
        \> x := \sfrac{1}{x}
        \end{align*}
        which continuously decreases the value of variable $x$ during 1 second
        and then applies the (discrete) operation $x := \sfrac{1}{x}$.  Suppose
        as well that our initial state is the environment $\sigma$ defined by
        $x \mapsto 1$.  Then by an application of rule \textbf{(diff-stop)} one
        deduces that this program outputs the environment $x \mapsto 1 - t$ at
        every time instant $t < 1$. On the other hand, by an application of
        rules \textbf{(diff-skip)}, \textbf{(asg-err)}, and \textbf{(seq-skip)}
        one easily deduces that the evaluation of the program fails at every
        time instant $t \geq 1$, due to a division by $0$.

        Notably the fact that failure occurs only at the time instants $t \geq
        1$ is a fundamental difference with respect to the famous hybrid programming
        language detailed in~\cite{Platzer10}. In the \emph{op. cit.} the
        language was designed in the spirit of Kleene algebra, which in
        particular forces the previous program to be \emph{indistinguishable}
        from \emph{e.g.} the program $x := \sfrac{x}{0}$. Whilst such a feature
        could be desirable in some verification scenarios it is clearly
        unnatural in a simulation-based environment such as ours.

Let us continue unravelling prominent features of our semantics with another
example. Consider the following hybrid program, \begin{align*} & \progwhile{x
\not = 0}{ x' = -1  \> {\prog{\blue{for}}} \> \sfrac{x}{2} \, } \> \scomp \> x
:= \sfrac{1}{x} \end{align*} paired with the environment $x \mapsto 1$.  This
program is an instance of a so-called Zeno loop: \emph{viz.} the loop involved
unfolds \emph{infinitely} many times with the duration of each iteration
becoming shorter and shorter (see details \emph{e.g.} in
\cite{goncharov2020implementing}).  In this particular case it is
straightforward to check that the duration of the $i$-th iteration is given by
$\sfrac{1}{2^i}$, and thus that the total duration $\sum_{i = 1}^{\infty}
\sfrac{1}{2^i}$ of the loop will be $1$. Now, by applying the operational rules
in \cref{big_step} one can successfully evaluate the program at every time
instant $t < 1$ (intuitively because every such $t$ is reached in a
\emph{finite} number of iterations). The same is \emph{false} for time instant $t = 1$ since
such requires a complete unfolding of the loop which is of course
computationally unfeasible. Thus operationally the potential point of failure $x :=
\sfrac{1}{x}$ in the program above never occurs, as the Zeno loop makes it
impossible to actually reach this command in the evaluation. 
These infinite behaviours are bounded in Lince by manually setting limits on the total time and on the number of unfoldings of while-loops, adjustable for each program.
\end{example}

Next, the semantics in the aforementioned `small-step' style is given in the
form of a relation $\ssto$ that is defined inductively according to the rules
in \cref{small_step}. These rules follow an analogous reasoning to the ones in
\cref{big_step} so we refrain from repeating explanations. 

\newcommand{\prem}[1]{(\textit{if\/ }#1)}
\newcommand{\nline}{\vspace{-8mm}}
\begin{figure*}[h]
\begin{minipage}{1\textwidth}
\begin{flalign*}
\textbf{(asg$^\to$)}
&&
x := e \sep\sigma\sep t
\ssto
\skp \sep \sigma[ x \mapsto \sem{e}(\sigma) ] \sep t 
&&
\prem{\sem{e}(\sigma) \text{ defined}}
\end{flalign*} \nline
\begin{flalign*}
\textbf{(asg-err$^\to$)}
&&
x := e\sep\sigma\sep t 
\ssto
\err 
&&
\prem{\sem{e}(\sigma) \text{ undefined}}
\end{flalign*} \nline
\begin{flalign*}
\textbf{(diff-stop$^\to$)}
  &&{\vec{x}' = \vec{\ell} \>
  \until \> e} \sep\sigma\sep t 
    \ssto
    \stp\sep\sigma [\vec{x} \mapsto \phi(t)]\sep 0
  &&
  \prem{\sem{e}(\sigma) > t}
\end{flalign*} \nline
\begin{flalign*}
\textbf{(diff-skip$^\to$)}
&& {\vec{x}' = \vec{\ell} \>
  \until \> e} \sep\sigma\sep t 
    \ssto
    \skp\sep\sigma[\vec{x} \mapsto \sigma(t)]\sep t - \sem{e}(\sigma)
&&
\prem{\sem{e}(\sigma) \leq t}
\end{flalign*} \nline
\begin{flalign*}
\textbf{(diff-err$^\to$)}
&& {\vec{x}' = \vec{\ell} \>
  \until \> e} \sep\sigma\sep t 
    \ssto \err 
&&
\prem{\sem{e}(\sigma) \text{ undefined}}
\end{flalign*} \nline
\begin{flalign*}
\textbf{(if-true$^\to$)}
&&
\progife{b}{p}{q}\sep\sigma\sep t \ssto {\prog p} \sep\sigma\sep t
&&
\prem{\sem{b}(\sigma) =\mathtt{tt}}
\end{flalign*} \nline
\begin{flalign*}
\textbf{(if-false$^\to$)}
&&
\progife{b}{p}{q}\sep\sigma\sep t \ssto {\prog q}\sep\sigma\sep t
&&
\prem{\sem{b}(\sigma)=\mathtt{ff}}
\end{flalign*} \nline
\begin{flalign*}
\textbf{(if-err$^\to$)}
&&
\progife{b}{p}{q}\sep\sigma\sep t \ssto \err
&&
\prem{\sem{b}(\sigma) \text{ undefined}}
\end{flalign*} \nline
\begin{flalign*}
\textbf{(wh-true$^\to$)}
&&
\progwhile{b}{{\prog p}}\sep\sigma\sep t \ssto {\prog p} 
\scomp \progwhile{b}{{\prog p}}\sep\sigma\sep t
&&
\prem{\sem{b}(\sigma)=\mathtt{tt}}
\end{flalign*} \nline
\begin{flalign*}
\textbf{(wh-false$^\to$)}
&&
\progwhile{b}{{\prog p}}\sep\sigma\sep t \ssto \skp\sep\sigma\sep t
&&
\prem{\sem{b}(\sigma)=\mathtt{ff}}
\end{flalign*}\nline 
\begin{flalign*}
\textbf{(wh-err$^\to$)}
&&
\progwhile{b}{{\prog p}}\sep\sigma\sep t \ssto \err
&&
\prem{\sem{b}(\sigma) \text{ undefined} }
\end{flalign*} 
\begin{flalign*}
\lrule{(seq-stop$^\to$)}{{\prog p}\sep\sigma\sep t \ssto \stp\sep\sigma'\sep t'}{
  {\prog p}\scomp {\prog q}\sep\sigma\sep t  \ssto \stp\sep\sigma'\sep t'}
&&
\lrule{(seq-skip$^\to$)}{{\prog p}\sep\sigma\sep t \ssto \skp\sep\sigma'\sep t'}{
  {\prog p}\scomp {\prog q}\sep\sigma\sep t  \ssto { \prog q}\sep\sigma'\sep t'}
\end{flalign*}
\begin{flalign*}
\lrule{(seq-err$^\to$)}{{\prog p}\sep\sigma\sep t \ssto \err}{
  {\prog p}\scomp {\prog q}\sep\sigma\sep t  \ssto \err }
  \hspace{0.9cm}
&&
\lrule{(seq$^\to$)}{{\prog p}\sep\sigma\sep t \ssto {\prog p'}\sep\sigma'\sep t'}{
  {\prog p}\scomp {\prog q}\sep\sigma\sep t  \ssto {\prog p'};{\prog q}\sep\sigma'\sep t'} \quad \prem{{\prog p'} \neq \stp \textit{ and } {\prog p'} \neq \skp} 
&&
\end{flalign*} 
  \end{minipage}
  \caption{Extension of the small-step operational semantics 
        in~\cite{goncharov2020implementing} with the possibility of failure.}
  \label{small_step}
\end{figure*}
As detailed in Corollary~\ref{cor:det} our small-step semantics is
deterministic. This is of course a key property in what concerns its
implementation and subsequent use in Lince for simulating hybrid programs. The
corollary is based on the following theorem.

\begin{theorem}
  \label{thm:determ}
  For every program ${\prog p}$, environment $\sigma$, and time instant
  $t$ there is \emph{at most one} applicable reduction rule.
\end{theorem}
Let $\ssto[\star]$ be the transitive closure of the small-step relation $\ssto$
that was previously presented. Intuitively $\ssto[\star]$ represents an
evaluation of one or more steps according to the small-step semantics. If
${\prog p}, \sigma, t \ssto[\star]  v$ we call $v$ `non-terminal' whenever it
is of the form ${\prog p'}, \sigma', t'$ for some hybrid program ${\prog p'}$,
environment $\sigma'$, and time instant $t'$; we call $v$ `terminal' otherwise.

\begin{corollary}[Determinism]
   \label{cor:det}
  Consider a program ${\prog p}$, an environment $\sigma$, and a time instant
  $t$. If ${\prog p} \sep \sigma \sep t \ssto[\star] v$ and ${\prog p} \sep
  \sigma \sep t \ssto[\star] u$ with both $v$ and $u$ terminal then we have $v
  = u$. 
\end{corollary}
\begin{proof}
  Follows by induction on the number of reduction steps and
  Theorem\,\ref{thm:determ}.\end{proof}
Next we will show that the small-step semantics and its big-step counterpart
are indeed equivalent.  We will use the two following results for this effect.

\begin{lemma}\label{lem:progress}
Given a program ${\prog p}$, an environment $\sigma$ and a time instant ${t}$
\begin{enumerate}
  \item if\/ ${\prog p}\sep\sigma\sep t\ssto {\prog p'}\sep\sigma'\sep t'$ and 
${\prog p}'\sep\sigma'\sep t'\bsto\skp\sep\sigma''$ then 
${{\prog p}\sep\sigma\sep t\bsto\skp\sep\sigma''}$;
  \item if\/ ${\prog p}\sep\sigma\sep t\ssto {\prog p'}\sep\sigma'\sep t'$ and 
${\prog p}'\sep\sigma'\sep t'\bsto\stp\sep\sigma''$ then 
${{\prog p}\sep\sigma\sep t\bsto\stp\sep\sigma''}$;
\item if\/ ${\prog p}\sep\sigma\sep t\ssto {\prog p'}, \sigma' ,t'$ and 
${\prog p}'\sep\sigma'\sep t'\bsto \err$ then 
${{\prog p}\sep\sigma\sep t\bsto \err}$;
\end{enumerate} 
\end{lemma}

\begin{proof}
  Follows by induction over the rules concerning the small-step relation. 
\end{proof}

\begin{proposition}\label{prop:shift}
  For all program {\prog p} and {\prog q}, environments $\sigma$
  and $\sigma'$, and time instants $t$, $t'$ and $s$, if\/
  ${\prog p}\sep\sigma\sep t \ssto {\prog q}\sep\sigma'\sep t'$ then\/
  ${\prog p}\sep\sigma\sep t + s \ssto {\prog q}\sep\sigma'\sep t' + s$;
  and if\/ ${\prog p}\sep\sigma\sep t \ssto \skp\sep\sigma'\sep t'$
  then\/
  ${\prog p}\sep\sigma\sep t+s \ssto \skp\sep\sigma'\sep {t'+s}$.
  If ${\prog p}\sep \sigma \sep t \ssto \err$ then
        ${\prog p}\sep \sigma \sep t +s \ssto \err$ 
\end{proposition}

\begin{proof}
        Follows straightforwardly by induction over the rules concerning the
        small-step relation and the algebraic properties of addition.
\end{proof}

\begin{theorem}[Equivalence]
\label{thm:osem-eq}
The small-step semantics and the big-step semantics are related in the
following manner. Given a program $\prog{p}$, an environment $\sigma$ and a
time instant $t$
\begin{enumerate}
  \item ${\prog p}\sep\sigma\sep t\bsto\skp\sep\sigma'$ iff\/ 
${\prog p}\sep\sigma\sep t\ssto[\star]\skp\sep\sigma'\sep 0$;
  \item ${\prog p}\sep\sigma\sep t\bsto\stp\sep\sigma'$ iff\/ 
${\prog p}\sep\sigma\sep t\ssto[\star]\stp\sep\sigma'\sep 0$;
\item ${\prog p}\sep\sigma\sep t\bsto \err$ iff\/ 
${\prog p}\sep\sigma\sep t\ssto[\star]\err $.
\end{enumerate} 
\end{theorem}

\begin{proof}
  The right-to-left direction is obtained by induction over the length of the
  small-step reduction sequence using \Cref{lem:progress}.  The left-to-right
  direction follows by induction over the big-step derivations together with
  Proposition~\ref{prop:shift}.
\end{proof}

\section{An Improved Simulator for Hybrid Programs}
\label{sec:interpreter}

This section summarises several improvements made to Lince's simulator of
hybrid programs since its original
publication~\cite{goncharov2020implementing}. These include (1) more expressive
assignments and durations in differential statements (by virtue of the results in the
preceding section); (2) a more user-friendly program syntax (by means of
syntactic sugar); (3) more informative error messages; and (4) a numerical
solver of systems of ordinary differential equations.  In order to render our
summary more lively we complement it with a running example involving an RLC
circuit in series with an On-Off source. It is designed to stabilise voltage
across the capacitor in the circuit at a specific value.

\smallskip
\noindent
\textbf{Running example: RLC circuits and harmonic oscillation.}
We present in~\cref{fig:dho-lince} the simulation of an \textit{RLC circuit in
series} (RLCS). This simulation models an electric system composed of a
resistor, a capacitor, an inductor, and a power source connected in series. The
power source strategically switches on and off, as a way to stabilise voltage
across the capacitor at a target value (say, $10V$). Such systems are known to
yield interesting results that are practically relevant for energy storage
voltage control systems, which help to mitigate voltage imbalances that could
otherwise damage electronic equipment. More details about such circuits and
associated differential equations are available for example
in~\cite{en14040832,hasan2019application}.  We present in~\cref{fig:dho-lince}
two variations of an RLCS circuit: one in which the capacitor voltage is in an
underdamped regime  -- with a resistance $\texttt{\small rU}$ of $0.5 \Omega$,
a capacitance $\texttt{\small c}$ of $0.047 F$,  and an inductance
$\texttt{\small l}$ of $0.047H$ -- and one in which the capacitor voltage is in
an overdamped regime -- with a resistance $\texttt{\small rO}$ of $4\Omega$ and
the same values as before for the capacitance and inductance.  The general idea
of our program is that the associated controller will read the voltage across
the capacitor (variable $\prog{under}$ for the underdamped case, $\prog{over}$
for the overdamped one) every 0.01 seconds, and set the voltage at the
source either to $0$ (off) or $18V$ (on) depending on the value read.
\begin{figure}[htb!]
    \centering
    \begin{minipage}{0.42\textwidth}
\begin{lstlisting}[style=lince]
under:=0; dU:=0; vU:=0; rU:=0.5;
over:=0;  dO:=0; vO:=0; rO:=4;
c:=0.047; l:=0.047;

while true do {
  if (under<10) then vU:=18;
                else vU:=0; 
  if (over<10)  then vO:=18;
                else vO:=0; 
  under'=dU, over'=dO, 
  dU'=-(dU*rU/l)
      -under/(l*c)+vU/(l*c),
  dO'=-(dO*rO/l)
      -over/(l*c)+vO/(l*c)
  for 0.01;
}
\end{lstlisting}
  \end{minipage}
  ~~~
  \wrap{\includegraphics[width=0.54\textwidth]{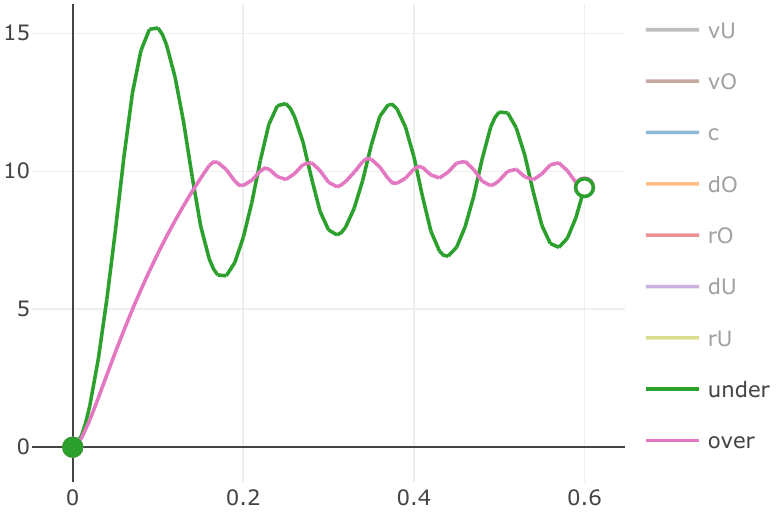}}
  \caption{Hybrid program (left) and its plot (right) of two variations
  of an RLC circuit that tries to maintain the voltage in the 
capacitor at $10V$.}
  \label{fig:dho-lince}
\end{figure}

\smallskip
\noindent
\textbf{Improvement's summary.}
The program just described is highly problematic for the original version of
Lince. This is due to two fundamental reasons related to the ODEs involved:
specifically (1) the equations used in the ODEs violate the linearity condition
presented in~\cref{sec:modelling} (they include variable multiplications); and
(2) the original solver of ODEs, mentioned in the introduction, fails to
produce solutions after few iterations, due to the sheer, exponential growth of
the involved expressions' size. We detail these issues and others next.

\smallskip
\noindent
\emph{Richer expressions.}
As illustrated in the introduction and in the previous RLCS example, there are
several essential, non-linear operations that are necessary to accomodate if
one wishes to employ Lince in the analysis of diverse, common hybrid scenarios.
We therefore now permit non-linear expressions outside of ODEs, essentially by
using as basis the grammar of hybrid programs that was described
in~\cref{sec:modelling}. Thus expressions outside the ODEs can now include for
example the operations: division and multiplication of variables,
more complex mathematical functions (such as square root extraction,
exponentials, logarithms, minimum/maximum, and (co)sine), and mathematical
constants (namely pi and Euler's constant).

As for expressions inside ODEs, the linearity constraint is kept but the
associated parser is much less rigid. A core feature is that it now tries to
convert non-linear expressions into equivalent linear ones via algebraic laws.
For example, it converts the expression $x \cdot 5$, which syntactically is not a linear expression, into the linear
one $5 \cdot x$ since multiplication is commutative. Most notably, it converts
non-linear expressions $x \cdot y$ into scalar multiplications $s \cdot x$ or $s \cdot y$
if it can infer that either $x$ or $y$ is a constant with value $s$. Such a feature is
critical in our RLCS~example, where we multiply variables in the
respective ODEs.

\smallskip
\noindent
\emph{More informative error messages.}
Several errors were undetected at an early stage of the simulation process,
which resulted in unintelligible error messages in many situations. We thus
added and improved the detection and notification of several key errors
occurring in typical usages of Lince, including when:
 (1) a partial function fails (such as in division by 0);
 (2) a variable is not properly initialised;
 (3) the number of arguments of a function is incorrect;
 (4) the solver fails to solve a system of ODEs; and
 (5) ODEs contain non-linear expressions after de-sugaring.
For example, in our RLCS simulation when defining \code{c} to be 0  we now
obtain the error ``\emph{Error: the divisor of the division 'rU/(c)' is
zero.}''.  In our experience, this more precise detection and notification of
errors drastically improved user experience.

\smallskip
\noindent
\emph{Numerical solver.}
As already mentioned, several hybrid programs such as our RLCS example cannot
be properly handled by the (exact) solver of ODEs (\emph{viz.}
SageMath~\cite{sage}) used by Lince. We have therefore implemented an
alternative, numerical solver based on the popular fourth-order Runge-Kutta
method.  At the theoretical level, this only required a
small adaptation of the operational semantics presented
in~\cref{sec:modelling}. Specifically we no longer assume that the solution
$\phi^{\vec{x}' = \vec{\ell}}_\sigma$ associated to a system of ODEs $\vec{x}'
= \vec{\ell}$ and an initial condition $\sigma$ is exact. At the
practical level, this allowed us to keep the size of expressions involved
in computations highly manageable thus allowing Lince to cover a broader range of
examples such as the RLCS.

\section{An Improved Visualiser for Hybrid Programs}
\label{sec:visualiser}

Many hybrid programs cannot be easily understood by simply plotting values of
variables over time. For example, in some cases one may wish to analyse the
movement of a vehicle in a 2D plane, or to analyse how its behaviour
varies due to changes in its initial position and velocity.  This section
presents an extension of Lince's visualisation capabilities in these two
directions. In the same spirit of the preceding section, we complement our
description with a running example.

\smallskip
\noindent
\textbf{Running example: avoiding and manoeuvring around obstacles.}
The \emph{Automatic Emergency Braking} (AEB) system is an autonomous driving
device that after reading its distance to an obstacle and its current velocity,
decides whether to decelerate until stopping~\cite{AEB_resume}. Here we present
a more advanced version of the AEB that after stopping also manoeuvres around
the obstacle -- clearly a process involving two or even three spatial
dimensions.  Such a system is called \emph{Automatic Emergency Braking with an
Overtaking Manoeuvre} (AEBOM).

The continuous dynamics of the AEBOM (\emph{i.e.} the differential equations
involved) is typically given by Dubins dynamics which essentially describe the
object's orientation over time (an angle) and its effect on the object's
velocity along the different spatial dimensions~\cite{Platzer3}. We adopt this
approach as well. For simplicity we additionally assume that our object is a
robot that is able to rotate around itself. The overall process of our AEBOM is
thus as follows: move forward until detecting the obstacle and in which case
decelerate until stopping; then rotate to the left and move forward a
prescribed number of meters (that depends on the obstacle's size); then rotate
right and move forward again a prescribed number of meters; and finally repeat
the last step.

\Cref{fig:aebom-vis} depicts the original visualisation of the AEBOM simulation
on the left, and a customised 2D visualisation that uses our extension on the
right. The respective implementation of the AEBOM, included in Lince online,
is not relevant to show at this stage, because our focus is at the moment on
describing new visualisation mechanisms and not features concerning code.
Observe as well that the plot on the right provides novel insights with respect
to the one on the left: whilst in the right it is clear that the robot does
not collide with the obstacle and performs the overtaking manoeuvre safely, in
the left it is much harder to see that the same occurs. We provide more
details about our improved plotting system next.
\begin{figure}[htb!]
    \centering
    \includegraphics[height=40mm]{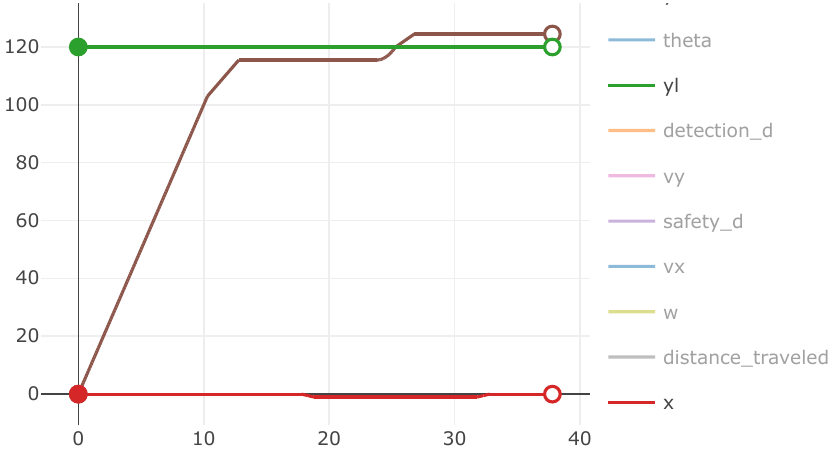}~~~~~~~~    \includegraphics[height=40mm]{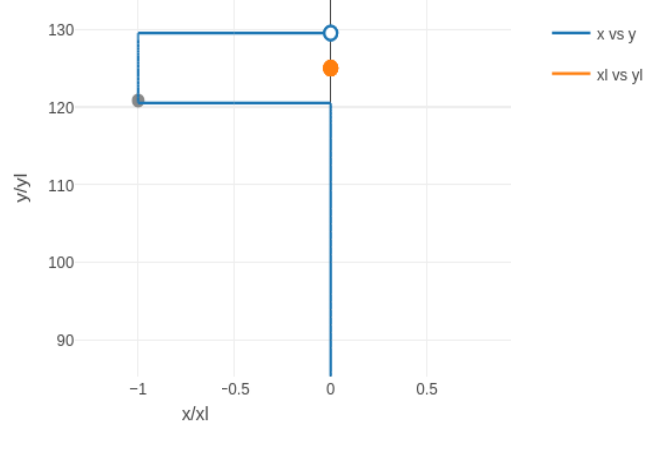}
    \caption{Plot of AEBOM using the traditional plotting system in Lince
    (left), and a new customised 2D plot (right) relating \code{x} with \code{y} (the
        robot's coordinates) and \code{xl} with \code{yl} (the
        obstacle's coordinates).}
        \label{fig:aebom-vis}
\end{figure}

\smallskip
\noindent
\textbf{Higher-dimensional trajectories and beyond.}
Our new visualisation framework in Lince uses the Plotly JavaScript library to
display plots\footnote{\url{http://plotly.com/}}. Among other things, we now
support 2D and 3D scatter plots, and include dedicated markers such as the
large circles indicating the start and end points of trajectories. When hovering over these markers, extra information
is displayed, \emph{e.g.} the respective values, relevant information about the
conditionals involved, and potential warnings. We also exploit the animation
functionality of Plotly in plots that do not include the time component, by
moving a highlighting circle through the trajectories capturing how values vary
throughout time. This feature is active by default. To take all these
possiblities into account, Lince allows the user to adjust different settings
of the plot under analysis so that she can obtain the best possible
configuration for her needs. We very briefly detail such settings next: 
\begin{itemize}
\item \emph{Axis}: Allows defining the relationships between variables which
        will automatically be presented in the respective plots. For example,
        by setting \code{[x, y, v]}, if the graph type is scatter, three
        separate graphs will be generated where the vertical axis represents
        each of the variables \code{x}, \code{y}, and \code{v}, while the
        horizontal axis represents time.  Choosing which variables to map to
        the axes is crucial for proper data analysis, allowing direct visual
        comparisons between different variables over time or with each other.
\item \emph{Max Time}: Refers to the duration of the
        simulation.
\item \emph{Max Iterations}: Specifies the maximum number of iterations (in
        while-loops) that the simulation can perform.
\item \emph{Graph Type}: Defines the type of graph to be used for visualising
        the simulation data, by selecting from the available types (`scatter'
        or `scatter3d'). In a nutshell, a scatter plot is a 2D graph used to
        display the relationship between two variables, with data points
        plotted in the two-dimensional plane. Scatter3D serves the same purpose
        but involves three variables, with data points plotted in the
        three-dimensional space.
\end{itemize}

The summarised settings are presented in~\cref{fig:configuration}, where the
values there listed are the ones used to obtain the plot in
~\cref{fig:aebom-vis} on the right.
\begin{figure}[htb!]
    \centering
    \includegraphics[width=0.29\textwidth]{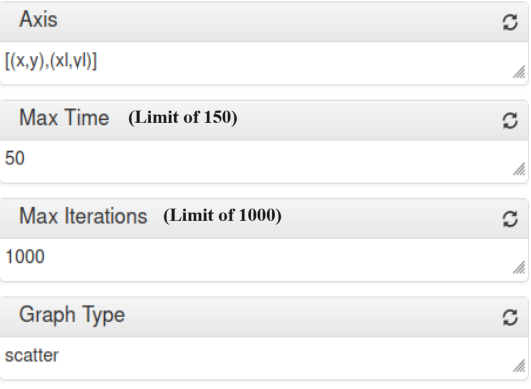}
    \caption{Input boxes that allow for the configuration of the visualisation.}
    \label{fig:configuration}
\end{figure}

\smallskip
\noindent
\textbf{Variability of initial conditions.}
As mentioned before, it is highly relevant take into account how the behaviour
of a hybrid program varies due to changes in its initial conditions. In the
AEBOM previously described in particular, it is of fundamental importance to
understand how the robot manoeuvres around an obstacle with respect to different
initial positions and velocities -- for it is unrealistic to expect that it
moves with well-known, exact conditions. A similar, more general discussion
can be consulted in~\cite{Platzer3}.

In order to address this aspect we extended Lince in two steps: first its
syntax now allows the listing of different initial conditions at the same time.
Such is illustrated in \cref{fig:aebom-multiple} on the left, with a snippet of
code used to specify initial values with respect to our robot in the AEBOM example.
The latter's initial position (\code{x},\code{y}) for example, can now be
either $(0,0)$, $(2,0)$, or $(4,0)$;  and similarly we have different
initial velocities (\code{vy}) towards the obstacle, $4$, $8$, and $12\,\sfrac{m}{s}$. Second Lince now pre-processes such listings in the code and
derives all possible combinations of initial conditions, which of course yields
several hybrid programs at once (in the standard syntax). These data is then
fed into Lince's visualiser which presents multiple simulations overlapped in
the same plot. Such is seen in~\cref{fig:aebom-multiple} on the right, again
with our AEBOM example, where we see that our robot behaves in the 
same way under different initial conditions.
\begin{figure}[htb!]
    \centering
    \begin{minipage}{0.25\textwidth}
\begin{lstlisting}[style=lince,mathescape=true]
//Initial Conditions
x :=[0,2,4];
y := 0;
vx:= 0;
vy:=[4,8,12];
xl:=[0,2,4];
yl:=[120,135,150];

...
\end{lstlisting}~
    \end{minipage}
    \wrap{\includegraphics[width=0.74\textwidth]{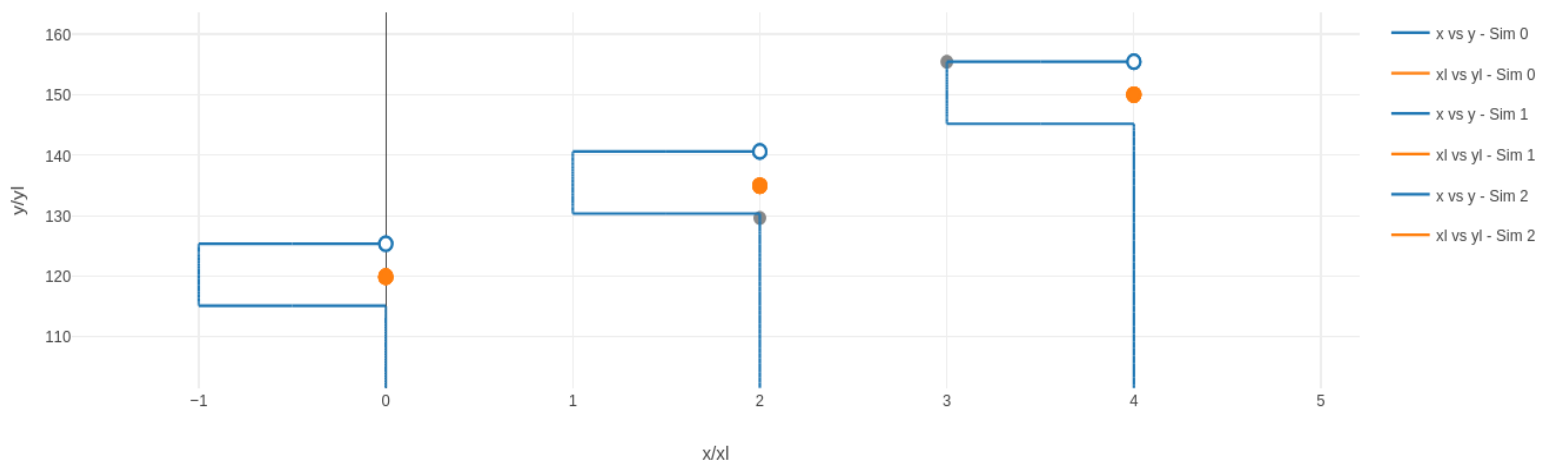}}
    \caption{Visualisation of multiple simulations overlapped concerning
            the AEBOM.
     }
    \label{fig:aebom-multiple}
\end{figure}

\section{Lince at Work: a Showcase of the Overall List of Improvements}
\label{sec:pursuit}

This section illustrates the overall list of improvements made to Lince (as
described in the preceding sections) working together in the design and
analysis of a complex hybrid scenario -- specifically we focus on a
multi-dimensional pursuit game between two players (for example two
drones)~\cite{manna93,anderson93,chaochen93,krilavicius05}.  Our illustration
focuses mainly on two aspects:
(1) Lince's capability tosimulate such scenarios,
with optimally configured 3D plotting systems; 
and
(2) the time that Lince takes to simulate increasingly larger systems,
to provide insights over limitations of the current implementation.

\smallskip
\noindent
\textbf{Pursuit Games.}
Pursuit games are a captivating class of problems involving multiple agents,
where at least one them (the pursuer) aims to capture or reach another (the
evader)~\cite{manna93,anderson93,chaochen93,krilavicius05,Platzer3}. Such games
are extensively studied across various disciplines, including mathematics, game
theory, robotics, and computer science, due to their practical and theoretical
significance. Indeed they model a wide range of real-world situations, from
military and security operations to animal behaviour and industrial
applications. 

In this section we explore a specific 3D pursuit game, where we perceive the
pursuer as a drone that attempts to capture another one in the
three-dimensional space.  This scenario is particularly challenging, due to the
additional complexity introduced by the third dimension which requires a higher
level of planning and coordination between the drones' movements. In order to
model this problem we base our game's continuous dynamics on Dubins
dynamics~\cite{Platzer3}, \emph{i.e.} as in~\cref{sec:visualiser} but now in
three dimensions. 

Our overarching strategy for the pursuer is to simply point its orientation to
the evader's position at every iteration in a certain while-loop. Of course
there are other options, such as that of (variations of) \emph{Dubins
paths}~\cite{Platzer3,351019}, but our version already suffices to properly
illustrate Lince at work. Technically our approach utilises the angular
velocity tensor to perform 3D infinitesimal rotations \cite{lehman}.
Additionally we use the cross product between the projection of the relative
velocity vector and the relative position vector in each plane to determine the
orientation of rotation among the three axes. We do not show here the coding
details of all these processes, since this is unnecessary for our illustration.
However the interested reader can consult details about these
in~\cite{351019,lehman}, and the complete code of our program is
included in the examples available in Lince online.

We now show the simulation of our game in Lince across different scenarios. In
the first case, the pursuer starts from the position \code{(300,300,600)} with
a velocity of \code{(-20,-10,0)}$\sfrac{m}{s}$, while the evader begins at the
position \code{(600,600,500)} with a velocity of
\code{(10,0,10)}$\sfrac{m}{s}$. The pursuer's angular velocity along each axis
is \code{(1/20)*2*pi()}$\sfrac{rad}{s}$ (20 seconds to complete a full
rotation); and for the evader \code{(1/40)*2*pi()}$\sfrac{rad}{s}$ (40 seconds
to complete a full rotation). The pursuer is allowed to actuate every 0.1s, and
it wins the game if it reaches a distance of less than one meter with respect to the
evader. Finally, for simplicity we assume a pre-defined set of movements for
the latter player.  Using these parameters, we simulated the corresponding
program in Lince and generated a 3D scatter plot of the positional variables
for both the pursuer and the evader, resulting in the graphical representation
shown in the \cref{pursuitGameSim1} after 73 seconds.
\begin{figure}[htb!]
    \centering    \wrap{\includegraphics[width=0.35\textwidth]{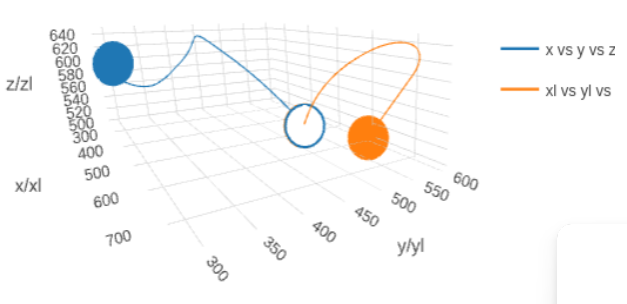}}
    \wrap{\includegraphics[width=0.35\textwidth]{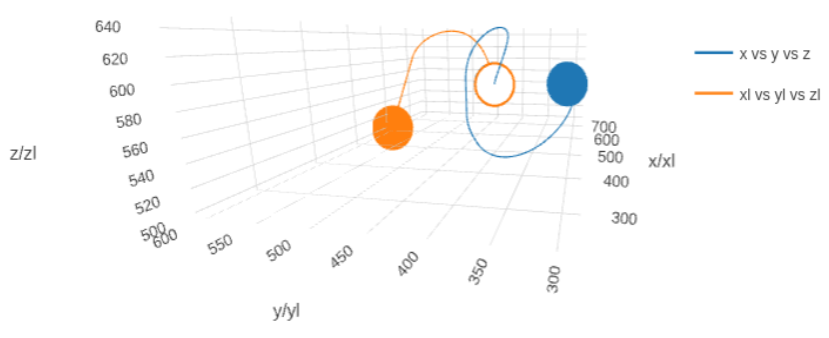}}        
    \caption{Two views of the same plot, where a pursuer (blue) captures an evader (orange).
     }
    \label{pursuitGameSim1}
\end{figure}

We can see that the decision strategy for the pursuer adopted in this hybrid
program successfully guided it to the evader, resulting in a capture at the
position \code{(691.26,441.92,561.12)} after 27.7 seconds.  However if we
change the initial velocity of the evader to a higher value, such as
\code{(20,0,9)}$\sfrac{m}{s}$, we no longer can visualise the capture of the evader
within the limits used for this simulation (\cref{pursuitGameSim2}). Indeed, Lince supports the customisation of bounds both on the maximal time and on the number of
times loops are unfolded, to avoid infinite computations. In this case, using larger bounds would allow the pursuer to capture the evader in the plot.
\begin{figure}[htb!]
    \centering    \wrap{\includegraphics[width=0.30\textwidth]{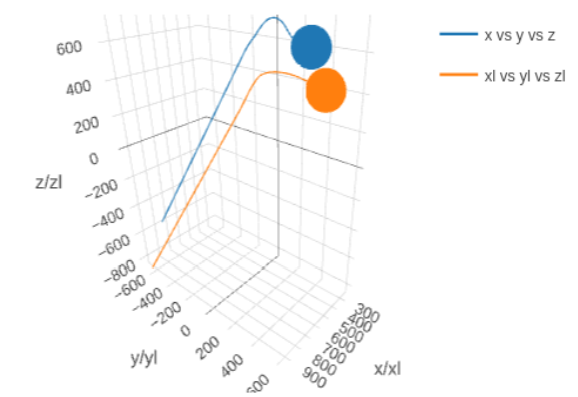}}
    \wrap{\includegraphics[width=0.30\textwidth]{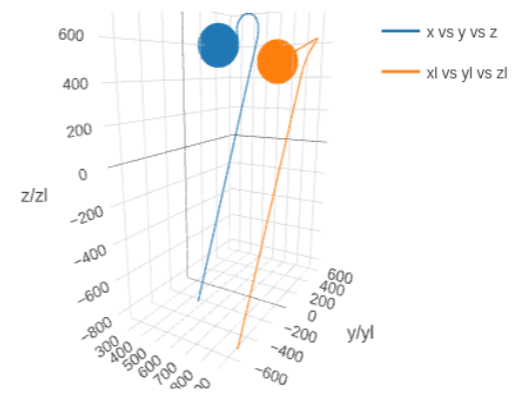}}        
    \caption{Similar plot to the one in \cref{pursuitGameSim1}, but using different initial velocities while keeping the same bounds on the size of the plot; this leaves out the point of the capture.
    }
    \label{pursuitGameSim2}
\end{figure}

Finally by taking advantage of the variability results presented
in~\cref{sec:visualiser} we very briefly study the effects of using different velocities
in this pursuit game. Specifically we adjust the angular velocity of the
pursuer along each axis to be either \code{(1/40)*2*pi()}$\sfrac{rad}{s}$ or
\code{(1/100)*2*pi()}$\sfrac{rad}{s}$, whilst keeping all other aspects. The resulting graphical representation
(after 220 seconds) is shown in~\cref{pursuitGameSim3}.
From the plots we observe that the pursuer successfully captures the evader
when the angular velocity is \code{(1/40)*2*pi()}$\sfrac{rad}{s}$ at the
position \code{(692.07,415.62,464.63)} in 34.8 seconds (left plot). However with an angular
velocity of \code{(1/100)*2*pi()}$\sfrac{rad}{s}$,
the pursuer does not capture the~evader in this time frame (middle plot).
These simulations showcase Lince's ability to model and simulate complex
scenarios, thus providing valuable insights into a system's behavior.

\begin{figure}[tb!]
    \centering 
    \wrap{\includegraphics[width=0.38\textwidth]{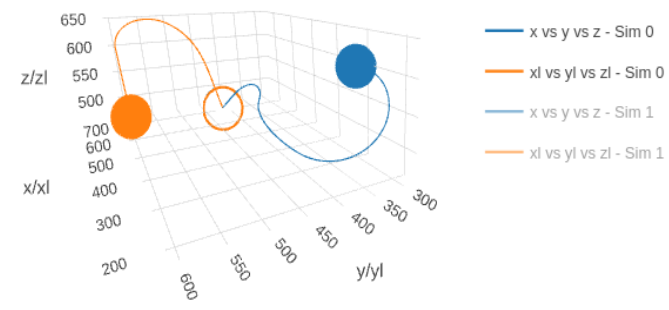}}  
    \wrap{\includegraphics[width=0.30\textwidth]{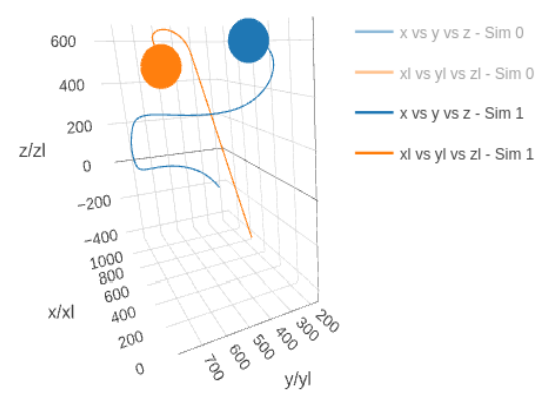}}  
    \wrap{\includegraphics[width=0.30\textwidth]{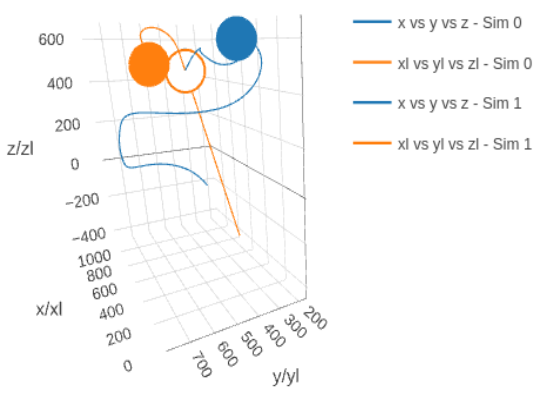}}
    \caption{Two simulations (left and middle) of a pursuit game using different initial velocities
      (\code{(1/40)*2*pi()}$\sfrac{rad}{s}$ and \code{(1/100)*2*pi()}, respectively);
      the right plot depicts both simulations overlaid.
    }
    \label{pursuitGameSim3}
\end{figure}

\smallskip
\noindent
\textbf{A brief overview of Lince's time performance.}
As shown in the previous example, Lince still has a few limitations concerning
performance. In order to give the reader a more concrete idea of them
we provide next an overview of how Lince fares perfomance-wise against the
examples presented in this paper.  First we need to give further context on how
Lince operates.

The first main observation is that now that Lince is equipped with an effective
numerical solver (recall~\cref{sec:interpreter}) it can operate in two starkly
different ways: one analytical with exact methods that rely on SageMath's
framework~\cite{sage}, the other numerical, based on progressively closer
approximations as described in~\cref{sec:interpreter}. Both operation modes
have significant differences performance-wise: most notably the former is
obviously slower and gives timeouts much more frequently than the latter
(recall our RLCS example in~\cref{sec:interpreter}). Interestingly the
bottleneck hinges not only on the employment of a precise solver, but also on
the fact that:
\begin{enumerate}
        \item this solver is external to Lince, specifically our tool needs to
                interact with a server, with all the usual delays that this
                implies;
        \item along the evaluation of a hybrid program, Lince needs to simplify
                resulting expressions over and over to make them tractable (due
                to them being symbolic and not numerical). 
\end{enumerate}
We saw first-hand in~\cref{sec:interpreter} how all these extra tasks running
behind the curtains inhibit Lince to simulate programs such as the RCLS
circuit.  The numerical solver, on the other hand, avoid these problems, but at
the cost of less precision which may have deep implications if one wishes to
have full guarantees that a simulation is correct, particularly if the system
at hand is chaotic~\cite{perko2013differential}. Needless to say, to find
methods that have the virtues of both approaches is a very interesting
challenge.

\newcommand{\colm}[1]{\textbf{\emph{#1}}\xspace}

\cref{table:time_analytics} lists several execution times of Lince against
different variations of the examples presented in the paper.
More specifically, each row represents one of our three examples with varying sampling times and total number of iterations. The example AEB is a variation of AEBOM, where the vehicle stops instead of performing an overtaking manoeuvre.  All these examples are fully available in our improved Lince online.

\begin{table}[tb!]
\centering
\caption{An overview of Lince's time performance with respect to the examples discussed
in this paper.  We consider different sampling times, number of
iterations, and both exact and approximate methods.}
\label{table:time_analytics}
\newcommand{\hd}[1]{\textbf{\begin{tabular}[c]{@{}c@{}}#1\end{tabular}}}
\begin{tabular}{cccccc}
\toprule
\textbf{} &
  \hd{Sampling \\ Time} &
  \hd{Nº of \\ Iterations} &
  \hd{Time\\ Symb-Server} &
  \hd{Time\\ Symb-Total} &
  \hd{Time\\ Numerical-Total} \\ \midrule
\multirow{3}{*}{\textbf{RLCS}}           & 0.01s & 1000 & - & - & 11.46s \\                                         & 0.1s & 1000 & - & - & 10.98s \\                                         & 1s & 150 & - & - & 1.14s \\ \midrule
\multirow{3}{*}{\textbf{AEB}}           & 0.01s & 184 & 23.56s & 23.70s & 0.41s \\                                        & 0.1s & 19          & 13.04s & 13.08s          & 0.18s \\                                        & 1s & 2          & 11.90s & 11.97s          & 0.14s \\ \midrule
\multirow{3}{*}{\textbf{AEBOM}}         & 0.01s & 1000 & - & - & 8.85s \\                                        & 0.1s  & 128          & - & -          & 0.62s \\                                        & 1s &  21         & - & -          & 0.35s \\ \midrule
\multirow{3}{*}{\textbf{Pursuit Games}} & 0.01s & 1000 & - & - & 66.60s \\                                        & 0.1s & 322          & - &  -         & 18.26s \\                                        & 1s & 150          & - &  -         & 7.85s \\ \bottomrule
\end{tabular}
\end{table}

We used a Linux laptop with a Intel quad-core i5 processor and 16GB RAM running both the server and the client.
The columns
\colm{Sampling time} and \colm{Nº of Iterations}
refer respectively to the rate at which
computational tasks need to be performed and the total number of times the
while-loop in the program involved is unfolded.
The column \colm{Time Symb-Total} presents the time since a new program is loaded, before parsing, until the plot is displayed in the browser.
The column \colm{Time Symb-Server} measures only the time taken since the launch of a dedicated process running SageMath until it is terminated at the end of a trajectory.
The column \colm{Time Numerical-Total} measures the time taken since a program is loaded until its plot is displayed, computed using numerical approximations.
Some observations over the values on \cref{table:time_analytics} follow below.
\begin{itemize}
    \item Most examples, except for AEB, reach a timeout (set in our server) when using the symbolic analysis, marked in the table with ``-''. The feasibility of AEB is mainly due to the smaller number of required calls to the symbolic engine.

    \item In the AEB example we observe that, when using exact methods, 
        around 99\
    \item The numerical mechanisms in the AEB example yield simulations
            significantly faster than in the exact counterpart. 

    \item The total time taken to numerically simulate the RLCS and AEBOM
            examples are shorter than in the Pursuit Games example. This is
            because these two examples involve fewer computations and the
            Pursuit Games use a 3D scatter plot, which is more
            computationally intensive than the 2D scatter plot.

    \item Larger sampling times imply reduced times in generating both the
            exact and numerical plots, due to the decreased number of
            computational operations.
            Consequently, it takes longer to simulate controllers with higher precision that actuate on physical processes such as movement, velocity, and time.
            However, many critical systems, e.g., in the context of autonomous driving and other embedded systems, may require such a high precision.

\end{itemize}

\section{Conclusion and Future Work}
\label{sec:conc}

We presented an improved version of Lince, which can now handle a broader class
of hybrid programs
and aims overall at improving user experience.  As previously discussed, this required
an extension with the possibility of failure of the operational semantics
introduced in~\cite{goncharov2020implementing}, the implementation of an
efficient numerical solver, and more informative error messages, among other
things.

We believe that our work opens up several research paths that we would like to
explore next. For example, thanks to the numerical solver it is now
straightforward to extend our language with non-linear differential equations,
which widens even more the range of programs that Lince can currently tackle.
Another interesting research path is the addition of probabilistic constructs
to Lince, such as measure sampling. We conjecture that this could be handled
easily in Lince via a random-number generator and part of the implemented
variability mechanisms that were presented in~\cref{sec:visualiser}. 

Yet another interesting research line is to connect Lince to the theorem prover
for hybrid programs KeYmaera X~\cite{Platzer3} -- specifically the connection
would consist of a suitable encoding from programs written in Lince to programs
written in KeYmaera X. Such would establish a workflow in which the engineer
first \emph{analyses} a given hybrid program via simulation mechanisms
(provided by Lince) and subsequently \emph{proves} properties about this program
(\emph{e.g.} correctness) via KeYmaera X.

\smallskip
\noindent
\textbf{Acknowledgments.}
This work is financed by National Funds through FCT - Fundação para a Ciência e
a Tecnologia, I.P. (Portuguese Foundation for Science and Technology) within
the project IBEX, with reference 10.54499/PTDC/CCI-COM/4280/2021.
This work is also partially supported by National Funds through FCT/MCTES, within the CISTER Unit (UIDP/UIDB/04234/2020); and by the EU/Next Generation, within the Recovery and Resilience Plan, within project Route 25 (TRB/2022/00061 -- C645463824-00000063).

\bibliographystyle{eptcs}
\bibliography{src/bib}

\end{document}